\documentclass{article}
\usepackage{fullpage}
 
\usepackage[utf8]{inputenc}
\usepackage[T1]{fontenc}
\usepackage{lmodern}
\usepackage{microtype}

\usepackage{color}
\definecolor{shadecolor}{rgb}{0.917969, 0.917969, 0.917969}
\usepackage{varioref}
\usepackage{framed}
\usepackage{amsmath}
\usepackage{amsthm}
\usepackage{amssymb}
\usepackage{graphicx}

\makeatletter

\theoremstyle{plain}

\newtheorem{thm}{\protect\theoremname}
  \theoremstyle{definition}
  \newtheorem{defn}[thm]{\protect\definitionname}
  \theoremstyle{plain}
  \newtheorem{lem}[thm]{\protect\lemmaname}
  \theoremstyle{definition}
  \newtheorem{problem}[thm]{\protect\problemname}
  \theoremstyle{plain}
  \newtheorem{fact}[thm]{\protect\factname}
  \theoremstyle{remark}
  \newtheorem{claim}[thm]{\protect\claimname}
  \theoremstyle{remark}
  \newtheorem{rem}[thm]{\protect\remarkname}
  \theoremstyle{plain}
  \newtheorem*{lem*}{\protect\lemmaname}

\makeatother

  \providecommand{\claimname}{Claim}
  \providecommand{\definitionname}{Definition}
  \providecommand{\factname}{Fact}
  \providecommand{\lemmaname}{Lemma}
  \providecommand{\problemname}{Problem}
  \providecommand{\remarkname}{Remark}
\providecommand{\theoremname}{Theorem}

\title{Counting matchings with $k$ unmatched vertices in planar graphs\footnote{This work was carried out while the author was a PhD student at Saarland University in Saarbr\"ucken, Germany, 
and while visiting the Simons Institute for the Theory of Computing in Berkeley, USA. 
The material also appears in his PhD thesis \cite{DBLP:phd/dnb/Curticapean15}.}
}

\author{Radu Curticapean\thanks{Institute for Computer Science and Control, Hungarian Academy of Sciences (MTA SZTAKI), Budapest, Hungary. The author is supported by ERC grant PARAMTIGHT, no.~280152.}}

\date{}

\begin{document}

\maketitle


\global\long\def\sharpP{\mathsf{\#P}}
\global\long\def\sharpETH{\mathsf{\#ETH}}

\global\long\def\FPT{\mathsf{FPT}}
\global\long\def\sharpWone{\mathsf{\#W[1]}}
\global\long\def\XP{\mathsf{XP}}

\global\long\def\pMatch{\mathsf{\#Match}}
\global\long\def\pClique{\mathsf{\#Clique}}

\global\long\def\pDefMatch{\mathsf{\#PlanarDefectMatch}}
\global\long\def\pRestrDefMatch{\mathsf{\#RestrDefectMatch}}
\global\long\def\pApexPerfMatch{\mathsf{\#ApexPerfMatch}}

\global\long\def\usat{\mathrm{usat}}
\global\long\def\PM{\mathcal{PM}}
\global\long\def\DM{\mathcal{DM}}
\global\long\def\M{\mathcal{M}}
\global\long\def\PerfMatch{\mathsf{\#PerfMatch}}
\global\long\def\MatchSum{\mathsf{\#MatchSum}}

\global\long\def\leqFptT{\leq_{\mathit{fpt}}^{T}}
\global\long\def\leqFptLin{\leq_{\mathit{fpt}}^{\mathit{lin}}}

\global\long\def\vec#1{\mathbf{#1}}

\begin{abstract}
We consider the problem of counting matchings
in planar graphs. While \emph{perfect} matchings in planar graphs
can be counted by a classical polynomial-time algorithm \cite{Kasteleyn1961,Temperley.Fisher1961,Kasteleyn1967},
the problem of counting all matchings (possibly containing unmatched
vertices, also known as \emph{defects}) is known to be $\sharpP$-complete
on planar graphs \cite{Jerrum1987}. 

To interpolate between the hard case of counting matchings and the easy case of counting perfect matchings,
we study the parameterized problem of counting matchings with exactly $k$
unmatched vertices in a planar graph $G$, on input $G$ and $k$.
This setting has a natural interpretation in statistical physics,
and it is a special case of counting perfect matchings in $k$-apex
graphs (graphs that can be turned planar by removing at most $k$ vertices).

Starting from a recent $\sharpWone$-hardness proof for
counting perfect matchings on $k$-apex graphs \cite{Curtican.Xia2015},
we obtain that counting matchings with $k$ unmatched vertices in planar graphs is
$\sharpWone$-hard.
In contrast, given a plane graph $G$ with $s$ distinguished faces,
there is an $\mathcal{O}(2^{s}\cdot n^{3})$ time algorithm for counting
those matchings with $k$ unmatched vertices such that all unmatched
vertices lie on the distinguished faces. This implies an $f(k,s)\cdot n^{O(1)}$
time algorithm for counting perfect matchings in $k$-apex graphs
whose apex neighborhood is covered by $s$ faces.{\small \par}
\end{abstract}

\section{Introduction}

The study of the computational complexity of counting problems was introduced in
a seminal paper by Valiant~\cite{Valiant1979a} that established
the class $\sharpP$ and proved counting perfect matchings in
an unweighted bipartite graph to be $\sharpP$-complete. In a companion
paper \cite{DBLP:journals/siamcomp/Valiant79}, he also showed that
counting all (not necessarily perfect) matchings in a graph is $\sharpP$-complete
as well. Even prior to these initial complexity-theoretic results,
problems related to matchings and perfect matchings played an
important role in various scientific disciplines.
For instance, the number of perfect matchings in a bipartite graph $G$ arises in
enumerative combinatorics and algebraic complexity as the \emph{permanent} \cite{Buergisser2000,Agrawal2006} of the bi-adjacency matrix associated with $G$.
In statistical physics, counting perfect matchings amounts to evaluating the \emph{partition
function} of the \emph{dimer model} \cite{Kasteleyn1967,Kasteleyn1961,Temperley.Fisher1961}:
The physical interpretation is that vertices are discrete points
that are occupied by atoms, and edges represent bonds between
the corresponding atoms. The partition function of $G$ is then essentially
defined as the number of perfect matchings in $G$, and it encodes
thermodynamic properties of the associated system. Likewise, the problem
of counting all matchings is known to statistical physicists as the
\emph{monomer-dimer model} \cite{Jerrum1987}; in this setting, some
points may be unoccupied by atoms. In the intersection of chemistry
and computer science, the number of matchings of a graph (representing
a molecule) is known as its \emph{Hosoya index} 
\cite{citeulike:4738899}.

In view of these applications and the $\sharpP$-hardness of counting (perfect) matchings, 
several relaxations were considered
to cope with these problems. Among these, \emph{approximate} counting and the restriction
to \emph{planar} graphs proved most successful. However, once we start
incorporating these relaxations, the seemingly very similar problems
of counting matchings and counting perfect matchings exhibit stark
differences:
\begin{itemize}
\item On planar graphs, perfect matchings can be counted in polynomial time
by the classical and somewhat marvelous FKT method \cite{Kasteleyn1967,Kasteleyn1961,Temperley.Fisher1961},
which reduces this problem to the determinant. The problem of counting
all matchings is however $\sharpP$-complete on planar graphs \cite{Jerrum1987}.
In particular, the algebraic machinery in the FKT method breaks down
for non-perfect matchings.
\item It was shown that the number of matchings in a graph admits a polynomial-time
randomized approximation scheme (FPRAS) on general graphs \cite{DBLP:journals/siamcomp/JerrumS89}.
By a substantial extension of this approach, an FPRAS for counting
perfect matchings in bipartite graphs was obtained \cite{Jerrum.Sinclair2004}
-- but despite great efforts, no FPRAS is known for general graphs.
\end{itemize}
In the present paper, we focus on the differing behavior of matchings
and perfect matchings on planar graphs. To this end, we study the
problem $\pDefMatch$ of counting matchings with $k$ unmatched vertices
(which we call $k$-defect matchings) in a planar graph $G$, on input
$G$ and $k$. This problem is clearly $\sharpP$-hard under Turing
reductions, as the $\sharpP$-hard number of matchings in $G$ can
be obtained as the sum of numbers of $k$-defect matchings in $G$
for $k=0,\ldots,|V(G)|$. On the other hand, $\pDefMatch$ can easily
be solved in time $|V(G)|^{\mathcal{O}(k)}$, as we can simply enumerate
all $k$-subsets $X \subseteq V(G)$ that represent potential defects,
count perfect matchings in the planar graph $G-X$ by the FKT method, and sum up 
these numbers.

\subsection{\label{sub:Parameterized-counting-problems}Parameterized counting
problems}

The fact that $\pDefMatch$ is $\sharpP$-hard and polynomial-time
solvable for constant $k$ suggests that this problem benefits from
the framework of \emph{parameterized counting complexity} \cite{Flum.Grohe2004}.
This area is concerned with \emph{parameterized} counting problems,
whose instances $x$ come with \emph{parameters}
$k$, such as $\pDefMatch$ or the problem $\pClique$ of counting
$k$-cliques in an $n$-vertex graph. Intuitively, the parameterized
problem $\pDefMatch$ considers $k$-defect matchings in planar graphs
with $k\ll n$, and the physical interpretation in terms of the monomer-dimer
model is that each configuration of the system admits a small
number of ``vacant'' points that are not occupied by atoms.

Note that both $\pDefMatch$ and $\pClique$ can be solved in time
$n^{\mathcal{O}(k)}$ and are hence in the so-called class $\XP$.
One important goal for such problems lies in finding algorithms with
running times $f(k)\cdot|x|^{\mathcal{O}(1)}$ for computable functions
$f$, which renders the problems \emph{fixed-parameter tractable}
(FPT) \cite{Flum.Grohe2004,Flum.Grohe2006}. If no FPT-algorithms
can be found for a given problem, one can try to show its $\sharpWone$-hardness.
This essentially boils down to finding a parameterized reduction from
$\pClique$, and it shows that FPT-algorithms for the problem would imply
FPT-algorithms for $\pClique$, which is considered unlikely.

For instance, to prove $\sharpWone$-hardness of $\pDefMatch$ by
reduction from $\pClique$, we would need to find an algorithm that
counts $k$-cliques of an $n$-vertex graph in time $f(k)\cdot n^{\mathcal{O}(1)}$
with an oracle for $\pDefMatch$. Additionally, the algorithm should
only invoke the oracle for counting $k'$-defect matchings with $k'\leq g(k)$.
Here, both the function $f$ appearing in the running time and the
blow-up function $g$ are arbitrary computable functions.

Furthermore, parameterized reductions can also be used to obtain lower
bounds under the exponential-time hypothesis $\sharpETH$, which postulates
that the satisfying assignments to formulas $\varphi$ in $3$-CNF
cannot be counted in time $2^{o(n)}$ \cite{Dell.Husfeldt2014,Impagliazzo.Paturi2001,Impagliazzo.Paturi2001b}.
For instance, it is known that $\pClique$ cannot be solved in time
$n^{o(k)}$ unless $\sharpETH$ fails \cite{Chen.Chor2005}. If we
reduce from $\pClique$ to a target problem by means of a reduction
that invokes only blow-up $\mathcal{O}(k)$, then $\sharpETH$ also
rules out $n^{o(k)}$ time algorithms for the target problem \cite{Lokshtanov.Marx2011}.

\subsection{Perfect matchings with planar-like parameters}

To put $\pDefMatch$ into context, let us survey some parameterizations
for the problem $\PerfMatch$ of counting perfect matchings and see
how these connect to $\pDefMatch$. 
\begin{itemize}
\item The FKT method for planar graphs was extended \cite{Galluccio.Loebl1998,regge2000combinatorial,Curtican.Xia2015}
from planar graphs to graphs of fixed genus $g$, resulting in $\mathcal{O}(4^{g}\cdot n^{3})$
time algorithms for $\PerfMatch$.
\item Polynomial-time algorithms for $\PerfMatch$ were obtained for $K_{3,3}$-free graphs
\cite{Little1974,Vazirani1989} and $K_{5}$-free graphs \cite{Straub.Thierauf2014}.
More generally, for every class of graphs excluding a fixed single-crossing
minor $H$ (that is, $H$ can be drawn in the plane with at most one
crossing), an $f(H)\cdot n^{4}$ time algorithm is known \cite{Curticape2014}.
\item A simple dynamic programming algorithm yields a running time of $3^t \cdot n^{\mathcal{O}(1)}$
for $\PerfMatch$ on graphs of treewidth $t$. By using fast subset convolution \cite{RBR09}, the running time can be improved to $2^{t}\cdot n^{\mathcal{O}(1)}$.
\end{itemize}
Since all of the tractable classes above exclude fixed minors for
fixed parameter values, one is tempted to believe that $\PerfMatch$
could be polynomial-time solvable on \emph{each }class of graphs excluding
a fixed minor $H$, and possibly even admit an FPT-algorithm when
parameterized by the minimum size of an excluded minor. This last possibility was however ruled
out by the following result:\footnote{In fact, recent unpublished work suggests the existence of constant-sized
minors $H$ such that $\PerfMatch$ is $\sharpP$-hard on $H$-minor
free graphs.} 
\begin{itemize}
\item $\PerfMatch$ is $\sharpWone$-hard on $k$-apex graphs \cite{Curtican.Xia2015}.
For $k\in\mathbb{N}$, a graph $G$ is $k$-apex if there is a set $A \subseteq V(G)$ of size $k$ such that $G-A$ is planar. The vertices in $A$ are called \emph{apices}. Since $k$-apex graphs exclude minors on $\mathcal{O}(k)$ vertices, the $\sharpWone$-hardness result for $\PerfMatch$ on
$k$-apex graphs implies $\sharpWone$-hardness of $\PerfMatch$ on
graphs excluding fixed minors $H$ (when parameterized by the minimum size of such an $H$).
\end{itemize}
Note that $\PerfMatch$ can be solved in time $n^{\mathcal{O}(k)}$
on $k$-apex graphs by brute-force in a similar way as $\pDefMatch$.
To cope with the $\sharpWone$-hardness of $\PerfMatch$ in $k$-apex
graphs and potentially obtain faster algorithms, we study two special
cases: 
\begin{enumerate}
\item We consider $\pDefMatch$, which is indeed a special case, as discussed
below.
\item We consider $\PerfMatch$ in $k$-apex graphs whose apices are adjacent
with only a bounded number of faces in the underlying planar graph.
More in Section~\ref{sub: few-faces} of the introduction.
\end{enumerate}

\subsection{\label{sub: from-apices}From $k$ apices to $k$ defects}

To count the $k$-defect matchings in a planar graph $G$, we can
equivalently count perfect matchings in the $k$-apex graph $G'$
obtained from $G$ by adding $k$ independent apex vertices adjacent
to all vertices of $G$: Every perfect matching of $G'$ then corresponds
to a $k$-defect matching of $G$, and likewise, every $k$-defect
matching of $G$ corresponds to precisely $k!$ perfect matchings
of $G'$. Hence $\pDefMatch$ reduces to $\PerfMatch$ on
$k$-apex graphs, even when the apices in these latter graphs form
an independent set and each apex is adjacent with all non-apex vertices.
Note that the $\sharpWone$-hardness for the general problem of $\PerfMatch$
on $k$-apex graphs does a priori not carry over to the special case
$\pDefMatch$, as the edges between apices and the planar graph cannot
be assumed to be complete bipartite graphs in the general problem.

Nevertheless, we show in Section~\ref{sec: Planar-k-defect-matchings} that $\pDefMatch$ is $\sharpWone$-hard.
To this end, we reduce from $\PerfMatch$ on $k$-apex graphs by means
of a ``truncated'' polynomial interpolation where we wish to recover
only the first $k$ coefficients from a polynomial of degree $n$.
The technique is comparable to that used in the first $\sharpWone$-hardness
proofs for counting matchings with $k$ edges \cite{Blaeser.Curtican2012,Curticape2013}.
Interestingly enough, our reduction maps $k$-apex graphs to instances
of counting $k$-defect matchings without incurring any parameter
blowup at all. In particular, we obtain the same almost-tight lower
bound under $\sharpETH$ that was known for $\PerfMatch$ on $k$-apex
graphs \cite{Curtican.Xia2015}.
\begin{thm}
\label{main thm: defect-match} The problem $\pDefMatch$ is $\sharpWone$-hard.
Furthermore, it admits no $n^{o(k/\log k)}$ time algorithm unless the counting exponential-time hypothesis $\sharpETH$
fails.
\end{thm}
It should be noted that the ``primal'' problem of counting matchings
with $k$ edges is $\sharpWone$-hard on general graphs \cite{Curticape2013,Curtican.Marx2014},
but becomes FPT on planar graphs \cite{Frick2004}. Furthermore, recall
that counting matchings with $0$ defects (that is, perfect matchings)
in general graphs is $\sharpP$-hard. See also Table~\ref{tab: complexities}
for the complexity of counting matchings in various settings.
\begin{table}
\begin{centering}
\begin{tabular}{c|c|c|}
\textbf{counting matchings} & on planar inputs & on general inputs\tabularnewline
\hline 
with $k$ edges & FPT by~\cite{Frick2004} & $\sharpWone$-complete by~\cite{Curticape2013,Curtican.Marx2014}\tabularnewline
\hline 
with $k$ defects & \textcolor{blue}{$\sharpWone$-hard by Thm.~\ref{main thm: defect-match}} & $\sharpP$-complete for $k=0$ by~\cite{Valiant1979a}\tabularnewline
\hline 
\end{tabular}
\par\end{centering}

\caption{\label{tab: complexities}Counting matchings under different parameterizations and input restrictions}
\end{table}

\subsection{\label{sub: few-faces}Few apices that also see few faces}

In Section~\ref{sec:Algorithm-for-restricted-apices}, we show that $\PerfMatch$ becomes
easier in $k$-apex graphs $G$ when the apex neighborhoods can all be covered
by $s$ faces of the underlying planar graph. This setting is motivated
by a structural decomposition theorem for graphs $G$ excluding a
fixed $1$-apex minor $H$: As shown in \cite{Demaine.Hajiaghayi2009},
based on \cite{Robertson.Seymour2003}, if $G$ excludes a fixed $1$-apex
minor $H$, then there is a constant $c_{H}\in\mathbb{N}$ such that
$G$ can be obtained by gluing together (in a formalized way) graphs
that have genus $\leq c_{H}$ after removing ``vortices'' from $\leq c_{H}$
faces and a set $A$ of $\leq c_{H}$ apex vertices, whose neighborhood
in $G-A$ is however covered by $\leq c_{H}$ faces. Our setting is
a simplification of this general situation as we forbid vortices,
gluing, and restrict the genus to $0$. We obtain an FPT-algorithm
for this restricted case:
\begin{thm}
\label{thm: apex}Given as input a graph $G$, a set $A\subseteq V(G)$
of size $k$ and a drawing of $G-A$ in the plane with $s$ distinguished faces
$F_{1},\ldots,F_{s}$
such that the neighborhood of $A$ is contained in the union of $F_{1},\ldots,F_{s}$,
we can count the perfect matchings of $G$ in time 
$2^{\mathcal{O}(2^{k}\cdot\log(k)+s)}\cdot n^{4}$.
\end{thm}
Note that even with $k=3$ and $s=1$, such graphs can have unbounded
genus, as witnessed by the graphs $K_{3,n}$ for $n\in\mathbb{N}$:
Each graph $K_{3,n}$ is a $3$-apex graph whose underlying planar
graph (which is an independent set) can be drawn on one single face. However, the genus of $K_{3,n}$
is known to be $\Omega(n)$ \cite{DBLP:books/daglib/0070576}.

To prove Theorem~\ref{thm: apex}, we first consider
a variant of $\pDefMatch$ where the input graph $G$ is given
as a planar drawing with $s$ distinguished faces.
The task is to count $k$-defect matchings such that
all defects are contained in the distinguished faces. This problem
is FPT, even when $k$ is not part of the parameter.
\begin{thm}
\label{thm: face-defects}Given as input a planar drawing of a graph
$G$ with $s$ distinguished faces $F_{1},\ldots,F_{s}$, the following
problem can be solved in time $\mathcal{O}(2^{s}\cdot n^{3})$: Count
the matchings in $G$ for which every defect is contained in $V(F_{1})\cup\ldots\cup V(F_{s})$.
\end{thm}
To prove Theorem~\ref{thm: face-defects}, we implicitly use the 
technique of combined signatures \cite{Curtican.Xia2015}: Using a
linear combination of two planar gadgets from \cite{Valiant2008},
we show that counting the particular matchings needed in Theorem~\ref{thm: face-defects}
can be reduced to $2^{s}$ instances of $\PerfMatch$ in planar graphs. 
We can phrase this result in a self-contained way that does not require
the general machinery of combined signatures.
It should be noted that the case $s=1$ was already solved by Valiant \cite{Valiant2008} and that our proof of Theorem~\ref{thm: face-defects} is a rather simple generalization of his construction. In a different context, this idea is also used in \cite{DBLP:journals/corr/Curticapean15a}.

More effort is then required to prove Theorem~\ref{thm: apex}, and we do so by reduction to Theorem~\ref{thm: face-defects}.
To this end, we label each vertex in the planar graph $G-A$ with
its neighborhood in the apex set $A$. Each $k$-defect matching in $G-A$ then
has a \emph{type}, which is the $k$-element multiset of $A$-neighborhoods
of its $k$ defects.\footnote{This resembles an idea from an algorithm for counting subgraphs
of bounded vertex-cover number \cite{Curtican.Marx2014}.} We will be able to count $k$-defect matchings $M$ of any specified
type among the $(2^{k})^{k}$ possible types, and we observe that
the number of extensions from $M$ to a perfect matching in $G$ depends
only on its type. This will allow us to recover the number of perfect matchings in $G$.

\section{\label{chap:Preliminaries}Preliminaries}

For $n\in\mathbb{N}$, write $[n]=\{1,\ldots,n\}$. 
Graphs $G$ are undirected and simple. 
They are unweighted unless specified otherwise. 
We write $N_{G}(v)$ for the neighborhood of $v\in V(G)$ in $G$.

\subsection{Polynomials}

We denote the degree of a polynomial $p\in\mathbb{Q}[x]$ by $\deg(p)$.
If $\vec{x}=(x_{1},\ldots,x_{t})$ is a list of indeterminates, then
we write $\mathbb{N}^{\vec{x}}$ for the set of all monomials over
$\vec{x}.$ A \emph{multivariate polynomial} $p\in\mathbb{Q}[\vec{x}]$
is a polynomial $p=\sum_{\theta\in\mathbb{N}^{\vec{x}}}a(\theta)\cdot\theta$
with $a(\theta)\in\mathbb{Q}$ for all $\theta\in\mathbb{N}^{\vec{x}}$,
where $a$ has finite support. 
The polynomial $p$ \emph{contains} a given monomial $\theta\in\mathbb{N}^{\vec{x}}$
if $a(\theta)\neq0$ holds. If $x$ is an indeterminate from $\vec{x}$,
then we write $\deg_{x}(p)$ for the \emph{degree of} $x$ in $p$.
This is the maximum number $k\in\mathbb{N}$ such that $p$ contains
a monomial $\theta$ with factor $x^{k}$. If $\vec{y}$ is a list of indeterminates, then we denote the \emph{total degree of} $\vec{y}$
in $p$ as the maximum degree of any monomial $\mathbb{N}^{\vec{y}}$
that is contained as a factor of a monomial in $p$.

Furthermore, if $p\in\mathbb{Q}[x,y]$ is a bivariate polynomial and
$\xi\in\mathbb{Q}$ is some arbitrary fixed value, we write $p(\cdot,\xi)$
for the result of the substitution $y\gets\xi$ in $p$, and we observe
that $p(\cdot,\xi)\in\mathbb{Q}[x]$. Likewise, we write $p(\xi,\cdot)$
for the result of substituting $x\gets\xi$.

\subsection{(Perfect) matching polynomials}

If $G$ is a graph, then a set $M\subseteq E(G)$ of vertex-disjoint edges
is called a matching. We write $\M[G]$ for the set of all matchings
of $G$. For $M\in\M[G]$, we write $\usat(M)$ for the set of unmatched
vertices in $M$. If $|\usat(M)|=k$ for $k\in\mathbb{N}$, we say
that $M$ is a $k$-defect matching, and we write $\DM_{k}[G]$ for
the set of $k$-defect matchings of $G$. We also write $\PM[G]=\DM_{0}[G]$
for the set of perfect matchings of $G$. 

If $G$ is an edge-weighted
graph with edge-weights $w:E(G)\to\mathbb{Q}$, then we define
\begin{equation}
\PerfMatch(G)=\sum_{M\in\PM[G]}\prod_{e\in M}w(e).\label{eq: PerfMatch}
\end{equation}
On planar graphs $G$, we can efficiently compute $\PerfMatch(G)$.
\begin{thm}
[\cite{Kasteleyn1961,Temperley.Fisher1961,Kasteleyn1967}]\label{thm: pm-planar-algo}
For planar edge-weighted graphs $G$, the value $\PerfMatch(G)$ can
be computed in time $\mathcal{O}(n^{3})$.
\end{thm}
If $G$ is a vertex-weighted graph with vertex-weights $w:V(G)\to\mathbb{Q},$
we define
\begin{equation}
\MatchSum(G)=\sum_{M\in\M[G]}\prod_{v\in\usat(M)}w(v).\label{eq: MatchSum}
\end{equation}
Both $\PerfMatch$ and $\MatchSum$ are also used in \cite{Valiant2008}. 
Note that zero-weights have different semantics in the two expressions: 
A vertex $v\in V(G)$ with $w(v)=0$ is required to
be matched in all matchings $M\in\M[G]$ that contribute a non-zero
term to $\MatchSum$. An edge $e\in E(G)$ with $w(e)=0$ can simply
be deleted from $G$ without affecting $\PerfMatch(G)$.

Finally, if $X$ is a formal indeterminate, we define the defect-generating
matching polynomial of unweighted graphs $G$ as 
\begin{equation}
\mu(G):=\sum_{M\in\M[G]}X^{|\usat(M)|}=\sum_{k=0}^{n}\#\DM_{k}[G]\cdot X^{k}.\label{eq: MatchGenerating}
\end{equation}
Note that $\mu(G)=\MatchSum(G')$ when $G'$ is obtained from $G$
by assigning weight $X$ to every vertex of $G$. In this paper, we
will be interested in the first $k$ coefficients of $\mu(G)$.
\begin{rem}
It is known \cite{Cai.Lu2009} that for every fixed $\xi\in\mathbb{Q}\setminus\{0\}$,
the problem of evaluating $\mu(G;\xi)$ on input $G$ is $\sharpP$-complete, even
on planar bipartite graphs $G$ of maximum degree $3$.
Note that the evaluation $\mu(G;0)$ counts the perfect matchings
of $G$.
\end{rem}

\subsection{Techniques from parameterized counting}

Please consider Section~\ref{sub:Parameterized-counting-problems}
for an introduction to parameterized counting complexity, and \cite{Flum.Grohe2004}
for a more formal treatment. We write $\leqFptT$ for parameterized (Turing) reductions between problems (as introduced in Section~\ref{sub:Parameterized-counting-problems}). Furthermore, we write $\leqFptLin$ for such parameterized reductions that incur only linear parameter blowup, i.e., on instances $x$ with parameter $k$, they only issue queries with parameter $\mathcal{O}(k)$.

Given a universe $\Omega$ and several ``bad'' subsets of $\Omega$,
the inclusion-exclusion principle allows us to count those elements
of $\Omega$ that avoid all bad subsets, provided that we know the
sizes of intersections of bad subsets.
\begin{lem}
\label{lem: incl-excl}Let $\Omega$ be a set and let $A_{1},\ldots,A_{t}\subseteq\Omega$.
For $\emptyset\subset S\subseteq[t]$, let $A_{S}:=\bigcap_{i\in S}A_{i}$
and define $A_{\emptyset}:=\Omega$. Then we have 
\[
\left|\Omega\setminus\bigcup_{i\in[t]}A_{i}\right|=\sum_{S\subseteq[t]}(-1)^{|S|}\left|A_{S}\right|.
\]
\end{lem}
In applications of Lemma~\ref{lem: incl-excl}, the left-hand side
of the equation corresponds to a quantity we wish to determine,
while the numbers $\left|A_{S}\right|$ for $S\subseteq[t]$ are computed by oracle calls.

We will also generously use the technique of polynomial interpolation: if a univariate polynomial $p$
has degree $n$ and we can evaluate $p(\xi)$ at $n+1$ distinct values
$\xi$, then we can recover the coefficients of $p$. This can be
generalized to multivariate polynomials: If $p$ has $n$ variables,
all of maximum degree $d$, and we are given sets $\Xi_{1},\ldots,\Xi_{n}$,
all of size $d+1$, along with evaluations of $p(\xi)$ on all grid
points $\xi\in\Xi_{1}\times\ldots\times\Xi_{n}$, then we can determine
the coefficients of $p$ in time $\mathcal{O}((d+1)^{3n})$.
\begin{lem}[\cite{DBLP:conf/icalp/Curticapean15}]
\label{lem: multivar-interpolation}Let $p\in\mathbb{Z}[x_{1},\ldots,x_{n}]$
be a multivariate polynomial, and for $i\in[n]$, let the degree of
$x_{i}$ in $p$ be bounded by $d_{i}\in\mathbb{N}$. Let $\Xi=\Xi_{1}\times\ldots\times\Xi_{n}\subseteq\mathbb{Q}^{n}$
with\textup{ $|\Xi_{i}|=d_{i}+1$ for all $i\in[n]$.} Then we can
compute the coefficients of $p$ with $\mathcal{O}(|\Xi|^{3})$ arithmetic
operations when given as input the set $\{(\xi,p(\xi))\mid\xi\in\Xi\}$.\end{lem}

\section{\label{sec: Planar-k-defect-matchings}Hardness of $\protect\pDefMatch$}

We now prove Theorem~\ref{main thm: defect-match}:
Given a \emph{planar} graph $G$ and $k\in\mathbb{N}$,
it is $\sharpWone$-hard to count the $k$-defect matchings of $G$.
This amounts to computing the coefficient of $X^{k}$ in the matching-defect
polynomial $\mu(G)$. We start from the $\sharpWone$-hardness for
the following problem $\pApexPerfMatch$, which follows from Theorem~1.2 and Remark~5.6 in \cite{Curtican.Xia2015}:
\begin{thm}[\cite{Curtican.Xia2015}]
\label{thm: ApexPerfMatch is hard}The following problem $\pApexPerfMatch$
is $\sharpWone$-hard: Compute the value of $\PerfMatch(G)$, when given as input an unweighted graph $G$ and
an independent set $A\subseteq V(G)$ of size $k$ such that $G-A$ is planar and
each vertex $v\in V(G)\setminus A$ satisfies $|N_{G}(v)\cap A|\leq1$. 
The parameter in this problem is $k$. Furthermore, assuming $\sharpETH$, the problem cannot be solved in time $n^{o(k / \log k)}$.
\end{thm}
In the proof of Theorem~\ref{main thm: defect-match}, we introduce
an intermediate problem $\pRestrDefMatch$:
\begin{problem}
The problem $\pRestrDefMatch$ is defined as follows: Given as input
a triple $(G,S,k)$ where $G$ is a planar graph, $S\subseteq V(G)$ is a set of vertices,
and $k\in\mathbb{N}$ is an integer, count those $k$-defect matchings of $G$ whose
defects all avoid $S$, i.e., those $k$-defect matchings $M$ with
$S\cap\usat(M)=\emptyset$. The parameter is $k$.
\end{problem}
The problem $\pRestrDefMatch$ is equivalent (up to multiplication by a simple factor) to the problem $\pApexPerfMatch$
on graphs $G$ whose apices $A$ are all adjacent to a common subset
$S$ of the planar graph $G-A$, and to no other vertices. Our overall
reduction then proceeds along the chain
\begin{equation}
\pApexPerfMatch\leqFptLin\pRestrDefMatch\leqFptLin\pDefMatch.\label{eq:reduction-chain}
\end{equation}

\subsection{From $\pApexPerfMatch$ to $\pRestrDefMatch$}
The first reduction in (\ref{eq:reduction-chain}) follows from an application of the inclusion-exclusion principle.
\begin{lem}
\label{lem: hardness of restricted dual matchings}We have $\pApexPerfMatch\leqFptLin\pRestrDefMatch$.
\end{lem}
\begin{proof}[Proof of Lemma~\ref{lem: hardness of restricted dual matchings}]
We reduce from $\pApexPerfMatch$ and wish to count perfect matchings
in an unweighted graph $G$ with apex set $A=\{a_{1},\ldots,a_{k}\}$
and planar base graph $H=G-A$. Note that $A$ is part of
the input, and it is an independent set. Furthermore, by definition
of $\pApexPerfMatch$, the set $V(H)$ admits a partition into $V_{1}\cup\ldots\cup V_{k}\cup W$
such that all vertices $v\in V_{i}$ for $i\in[k]$ are adjacent to
the apex $a_{i}$ and to no other apices, while no vertex $v\in W$
is adjacent to any apex.
In other words, each vertex $v\in V(H)$ can be colored by its unique
adjacent apex, or by a neutral color if $v\in W$. 

Recall that $\DM_{k}[H]$
denotes the set of $k$-defect matchings in $H$. We call a $k$-defect
matching $M\in\DM_{k}[H]$ \emph{colorful} if $|\usat(M)\cap V_{i}|=1$
holds for all $i\in[k]$, and we write $\mathcal{C}$ for the set
of all such $M$. Note that $\usat(M)\cap W=\emptyset$ for $M\in\mathcal{C}$,
since none of its $k$ defects are left over for $W$.

We claim that $\PM[G]\simeq\mathcal{C}$: If $M\in\PM[G]$, then $N=M-A$
satisfies $N\in\mathcal{C}$. Conversely, every $N\in\mathcal{C}$
can be extended to a unique $M\in\PM[G]$ by matching the unique $i$-colored
defect to its unique adjacent apex $a_{i}$.

Given oracle access to $\pRestrDefMatch$, we can determine $\#\mathcal{C}$
by the inclusion-exclusion principle from Lemma~\ref{lem: incl-excl}:
For $i\in[k]$, let $\mathcal{A}_{i}$ denote the set of those $M\in\DM_{k}[H]$
whose defects avoid color $i$, i.e., they satisfy $\usat(H,M)\cap V_{i}=\emptyset$.
Then 
\[
\mathcal{C}=\DM_{k}[H]\setminus\bigcup_{i\in[k]}\mathcal{A}_{i}.
\]
For $S\subseteq[k]$, write $\mathcal{A}_{S}=\bigcap_{i\in S}\mathcal{A}_{i}$.
We can compute $\#\mathcal{A}_{S}$ by an oracle call
to $\pRestrDefMatch$ on the instance $(H,\bigcup_{i\in S}V_{i},k)$,
so we can compute $\#\mathcal{C}=\#\PM[G]$ via inclusion-exclusion
(Lemma~\ref{lem: incl-excl}) and $2^{k}$ oracle calls to $\pRestrDefMatch$.
\end{proof}

\subsection{From $\pRestrDefMatch$ to $\pDefMatch$}
For the second reduction in (\ref{eq:reduction-chain}), we wish to
solve instances $(G,S,k)$ to $\pRestrDefMatch$ when given only an
oracle for counting $k$-defect matchings in planar graphs, \emph{without}
the ability of specifying the set $S$. Let $G$, $S$ and $k$ be
fixed in the following. Our reduction involves manipulations on polynomials,
such as a truncated version of polynomial division:
\begin{lem}
\label{lem:PolyDiv}Let $X$ be an indeterminate, and let $p,q\in\mathbb{Z}[X]$
be polynomials $p=\sum_{i=0}^{m}b_{i}X^{i}$ and $q=\sum_{i=0}^{n}a_{i}X^{i}$
with $a_{0}\neq0$. For all $t\in\mathbb{N}$, we can compute $b_{0},\ldots,b_{t}$
with $\mathcal{O}(t^{2})$ arithmetic operations from $a_{0},\ldots,a_{t}$
and the first $t+1$ coefficients of the product $pq$.
\end{lem}
\begin{proof}
Let $c_{0},\ldots,c_{n+m}$ enumerate the coefficients of the product
$pq$. By elementary algebra, we have $c_{i}=\sum_{\kappa=0}^{i}a_{\kappa}b_{i-\kappa}$,
which implies the linear system 
\begin{equation}
\left(\begin{array}{ccc}
a_{0}\\
\vdots & \ddots\\
a_{t} & \ldots & a_{0}
\end{array}\right)\left(\begin{array}{c}
b_{0}\\
\vdots\\
b_{t}
\end{array}\right)=\left(\begin{array}{c}
c_{0}\\
\vdots\\
c_{t}
\end{array}\right).\label{M:eq: Linear system of polynomial coefficients}
\end{equation}
As this system is triangular with $a_{0}\neq0$ on its main diagonal,
it has full rank and can be solved uniquely for $b_{0},\ldots,b_{t}$
with $\mathcal{O}(t^{2})$ arithmetic operations.
\end{proof}

Our proof also relies upon a gadget which will allow to distinguish
$S$ from $V(G)\setminus S$.
\begin{defn}
For $\ell\in\mathbb{N}$, an $\ell$-rake $R_{\ell}$ is a matching
$M$ of size $\ell$, together with an additional vertex $w$ adjacent
to one vertex of each edge in $M$:

\begin{center}
\includegraphics[width=0.16\textwidth]{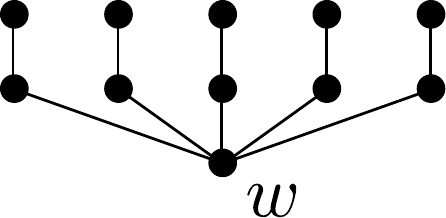}\vspace{-3mm}

\par\end{center}

Let $G_{S,\ell}$ be the graph obtained from attaching $R_{\ell}$
to each $v\in S$. This means adding a local copy of $R_{\ell}$ to
$v$ and identifying the copy of $w$ with $v$. Please note that
vertices $v\in V(G)\setminus S$ receive no attachments in $G_{S,\ell}$.
\end{defn}
It is obvious that $G_{S,\ell}$ is planar if $G$ is. Recall the
defect-generating matching polynomial $\mu$ from (\ref{eq: MatchGenerating}).
We first show that, for fixed $\ell\in\mathbb{N}$, the polynomial
$\mu(G_{S,\ell})$ can be written as a weighted sum over matchings
$M\in\M[G]$, where each $M$ is weighted by an expression that depends
on the number $|\usat(M)\cap S|$. Ultimately, we want to tweak these
weights in such a way that only matchings with $|\usat(M)\cap S|=0$
are counted.
\begin{lem}
\label{lem: weighted-matching-sum}Define polynomials $r,f_{\ell}\in\mathbb{Z}[X]$
and $s\in\mathbb{Z}[X,\ell]$ by 
\begin{eqnarray*}
r(X) & = & 1+X^{2}, \\
s(X,\ell) & = & \ell+1+X^{2}, \\
f_{\ell}(X) & = & (1+X^{2})^{|S|(\ell-1)}.
\end{eqnarray*}
Then it holds that 
\begin{equation}
\mu(G_{S,\ell},X)=f_{\ell}\cdot\sum_{M\in\M[G]}X^{|\usat(M)|}\cdot r^{|S\setminus\usat(M)|}\cdot s^{|S\cap\usat(M)|}.\label{eq: mu-weights}
\end{equation}

\end{lem}

\begin{proof}
Every matching $M\in\mathcal{M}[G]$ induces a certain set $\mathcal{C}_{M}\subseteq\mathcal{M}[G_{S,\ell}]$
of matchings in $G_{S,\ell}$, where each matching $N\in\mathcal{C}_{M}$
consists of $M$ together with an extension by rake edges. The family
$\{\mathcal{C}_{M}\}_{M\in\M[G]}$ is easily seen to partition $\mathcal{M}[G_{S,\ell}]$,
and we obtain
\begin{equation}
\mu(G_{S,\ell},X)=\sum_{M\in\M[G]}\underbrace{\sum_{N\in\mathcal{C}_{M}}X^{|\usat(N)|}}_{=:e(M)}.\label{eq: mu-part}
\end{equation}

Every matching $N\in\mathcal{C}_{M}$ consists of $M$ and rake edges,
which are added independently at each vertex $v\in S$. Hence, the
expression $e(M)$ in (\ref{eq: mu-weights}) can be computed from
the product of the individual extensions at each $v\in S$.
To calculate the factor obtained by such an extension, we have
to distinguish whether $v$ is unmatched in $M$ or not. The possible
extensions at $v$ are also shown in Figure~\ref{fig: rake states}.
\begin{figure}
\begin{centering}
\includegraphics[width=0.75\textwidth]{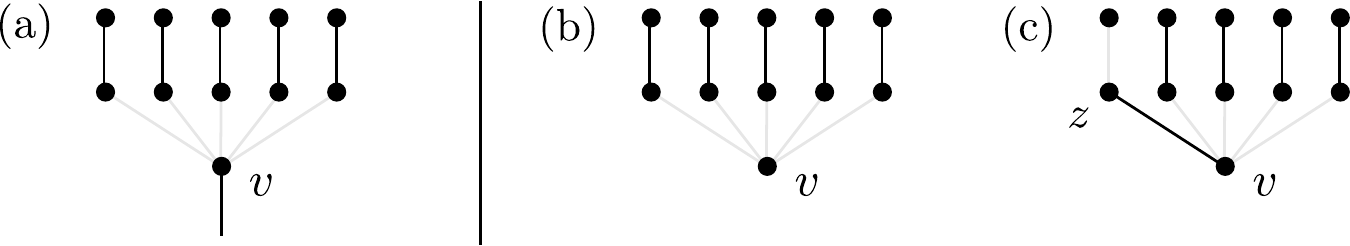}
\par\end{centering}

\caption{\label{fig: rake states}Possible types of extensions of the rake
at $v$. The left case corresponds to $v\protect\notin\protect\usat(M)$,
and the two right cases correspond to $v\in\protect\usat(M)$.}
\end{figure}

\begin{description}
\item [{$v\notin\usat(M):$}] We can extend $M$ at $v$ by any subset
of the $\ell$ rake edges not adjacent to $v$, as shown in Figure~\ref{fig: rake states}.a.
In total, these $2^{\ell}$ extensions contribute the factor $(1+X^{2})^{\ell}=(1+X^{2})^{\ell-1}r$.
\item [{$v\in\usat(M):$}] We have two choices for extending, shown in
the right part of Figure~\ref{fig: rake states}:
Firstly, we can extend as in the case $v\notin\usat(M)$, and then we obtain
the factor $X(1+X^{2})^{\ell}$. Here, the additional factor $X$
corresponds to the unmatched vertex $v$. This situation is shown
in Figure~\ref{fig: rake states}.b.
Secondly, we can match $v$ to one of its $\ell$ incident rake edges, say to
$e=vz$ for a rake vertex $z$, as in Figure~\ref{fig: rake states}.c.
Then we can choose a matching among the $\ell-1$ rake edges not incident
with $z$. This gives a factor of $\ell X(1+X^{2})^{\ell-1}$. Note
that $v$ is matched, but the vertex adjacent to $z$ is not, yielding
a factor of $X$.

In total, if $v\in\usat(M)$, we obtain the factor $X(1+X^{2})^{\ell}+\ell X(1+X^{2})^{\ell-1}=X(1+X^{2})^{\ell-1}s$. 

\end{description}

In each matching $N\in\mathcal{C}_{M}$, every unmatched vertex in
$\bar{S}=V(G)\setminus S$ contributes a factor $X$. By multiplying
the contributions of all $v\in V(G)$, we have thus shown that 
\begin{eqnarray*}
e(M) & = & f_{\ell}(X)\cdot X^{|\bar{S}\cap\usat(M)|}\cdot r^{|S\setminus\usat(M)|}\cdot(Xs)^{|S\cap\usat(M)|}\\
 & = & f_{\ell}(X)\cdot X^{|\usat(M)|}\cdot r^{|S\setminus\usat(M)|}\cdot s^{|S\cap\usat(M)|}
\end{eqnarray*}
and together with (\ref{eq: mu-part}), this proves the claim.
\end{proof}

Due to the factor $f_{\ell}$, the expression $\mu(G_{S,\ell})$ is
not a polynomial in the indeterminates $X$ and $\ell$. We define
a polynomial $p\in\mathbb{Z}[X,\ell]$ by removing this factor. 
\begin{equation}
p(X,\ell):=\sum_{M\in\M[G]}X^{|\usat(M)|}\cdot r^{|S\setminus\usat(M)|}\cdot s^{|S\cap\usat(M)|}.\label{eq:poly-p}
\end{equation}
Depending upon the concrete application, we will consider $p\in\mathbb{Z}[X,\ell]$
as a polynomial in the indeterminates $\ell$ and $X$, or as a polynomial
$p\in(\mathbb{Z}[\ell])[X]$ in the indeterminate $X$ with coefficients
from $\mathbb{Z}[\ell]$. In this last case, we write $p=\sum_{i=0}^{n}a_{i}X^{i}$
with coefficients $a_{i}\in\mathbb{Z}[\ell]$ for $i\in\mathbb{N}$
that are in turn polynomials. Then we define 
\begin{equation}
[p]_{k}:=\sum_{i=0}^{k}a_{i}X^{i}\label{eq:truncated-p}
\end{equation}
as the restriction of $p$ to its first $k+1$ coefficients. For later
use, let us observe the following simple fact about $[p]_{k}$, considered
as a polynomial $[p]_{k}\in\mathbb{Z}[X,\ell]$.
\begin{fact}
\label{fact: monomials of p}For $i,j\in\mathbb{N}$, every monomial
$\ell^{i}X^{j}$ appearing in $[p]_{k}$ satisfies $i\leq j\leq k$.\end{fact}
\begin{proof}
Recall $r$ and $s$ from Lemma~\ref{lem: weighted-matching-sum}.
The indeterminate $\ell$ appears in $s$ with degree $1$, but it
does not appear in $r$. In the right-hand side of (\ref{eq:poly-p}),
every term containing a factor $s^{t}$, for $t\in\mathbb{N}$, also
contains the factor $X^{t}$, because $|S\cap\usat(M)|\leq|\usat(M)|$
trivially holds. Hence, whenever $\ell^{i}X^{j}$ is a monomial in
$p$, then $i\leq j$. Since the maximum degree of $X$ in $[p]_{k}$
is $k$ by definition, the claim follows.
\end{proof}
In the next lemma, we show that knowing the coefficients of $[p]_{k}$
allows to solve the instance $(G,S,k)$ to $\pRestrDefMatch$ from
the beginning of this subsection. After that, we will show how to
compute $[p]_{k}$ with an oracle for $\pDefMatch$.
\begin{lem}
\label{lem: polynomial gives restricted dual matchings}Let $\mathcal{N}$
denote the set of (not necessarily $k$-defect) matchings in $G$
with $\usat(M)\cap S=\emptyset$. For all $k\in\mathbb{N}$, we can
compute the number of $k$-defect matchings in $\mathcal{N}$ in polynomial
time when given the coefficients of $[p]_{k}$.\end{lem}
\begin{proof}
For ease of presentation, assume first we knew \emph{all} coefficients
of $p$ rather than only those of $[p]_{k}$. We will later show how
to solve the problem when given only $[p]_{k}$. 

Starting from $p$, we perform the substitution
\begin{equation}
\ell\gets-(1+X^{2})\label{eq:subst}
\end{equation}
to obtain a new polynomial $q\in\mathbb{Z}[X]$ from $p$. By definition
of $s$ (see Lemma~\ref{lem: weighted-matching-sum}), we have 
\begin{equation}
s(X,-(1+X^{2}))=0,\label{eq: eq s =00003D0}
\end{equation}
so every matching $M\notin\mathcal{N}$ has zero weight in $q$. To
see this, note that by (\ref{eq:poly-p}), the weight of each matching
$M\in\M[G]$ in $p$ contains a factor $s^{|S\cap\usat(M)|}$. But
due to (\ref{eq: eq s =00003D0}), the corresponding term in $q$
is non-zero only if $|S\cap\usat(M)|=0$. We obtain 
\[
q=\sum_{M\in\mathcal{N}}X^{|\usat(M)|}\cdot(1+X^{2})^{|S\setminus\usat(M)|}.
\]
Since every $M\in\mathcal{N}$ satisfies $|S\setminus\usat(M)|=|S|$,
this simplifies to 
\begin{equation}
q=(1+X^{2})^{|S|}\cdot\underbrace{\sum_{M\in\mathcal{N}}X^{|\usat(M)|}}_{=:q'}\label{eq:def-poly-q}
\end{equation}
and we can use standard polynomial division by $(1+X^{2})^{|S|}$
to obtain
\begin{equation}
q'=q/(1+X^{2})^{|S|}.\label{eq:poly-div}
\end{equation}
By (\ref{eq:def-poly-q}), for all $k\in\mathbb{N}$, the coefficient
of $X^{k}$ in $q'$ counts precisely the $k$-defect matchings in
$\mathcal{N}$. This finishes the discussion of the idealized setting
when all coefficients of $p$ are known. Recall the three steps involved:
The substitution in (\ref{eq:subst}), the polynomial division in
(\ref{eq:poly-div}), and the extraction of the coefficient $X^{k}$
from $q'$.

The full claim, when only $[p]_{k}$ rather than $p$ is given, can
be shown similarly, but some additional care has to be taken. First,
we perform the substitution (\ref{eq:subst}) on $[p]_{k}$ rather
than $p$. This results in a polynomial $b\in\mathbb{Z}[X]$, for
which we claim the following:
\begin{claim}
\label{claim: bk =00003D qk}We have $[b]_{k}=[q]_{k}$.
\end{claim}

\begin{proof}
Let $\Theta_{\leq i}$ for $i\in\mathbb{N}$ denote the set of monomials
in $p$ with degree $\leq i$ in $X$. The substitution (\ref{eq:subst})
maps every monomial $\theta$ in the indeterminates $X$ and $\ell$
to some polynomial $g_{\theta}\in\mathbb{Z}[X]$. Writing $a(\theta)\in\mathbb{Z}$
for the coefficient of $\theta$ in $p$, we obtain $q,b\in\mathbb{Z}[X]$
with 
\begin{eqnarray}
q & = & \sum_{\theta\in\Theta_{\leq n}}a(\theta)\cdot g_{\theta},\label{eq:poly-q}\\
b & = & \sum_{\theta\in\Theta_{\leq k}}a(\theta)\cdot g_{\theta}.\label{eq:poly-r}
\end{eqnarray}
We can conclude that 
\begin{equation}
\left[q\right]_{k}\underset{\eqref{eq:poly-q}}{=}\left[\sum_{\theta\in\Theta_{\leq n}}a(\theta)\cdot g_{\theta}\right]_{k}=\left[\sum_{\theta\in\Theta_{\leq k}}a(\theta)\cdot g_{\theta}\right]_{k}\underset{\eqref{eq:poly-r}}{=}\left[b\right]_{k},\label{eq:q=00003Dr}
\end{equation}
where the second identity holds since, whenever $\theta$ has degree
$i$ in $X$, for $i\in\mathbb{N}$, then $g_{\theta}$ contains a
factor $X^{i}$. Hence, for $\theta\in\Theta_{\leq n}\setminus\Theta_{\leq k}$,
no terms of the polynomial $g_{\theta}$ appear in 
$\left[\sum_{\theta\in\Theta_{\leq n}}a(\theta)\cdot g_{\theta}\right]_{k}.$ 
\end{proof}

Recall the polynomial $q'$ from (\ref{eq:poly-div}); it remains
to apply polynomial division as in (\ref{eq:poly-div}) to recover
$[q']_{k}$ from $[b]_{k}$. To this end, we observe that the constant
coefficient in $(1+X^{2})^{|S|}$ is $1$, and that all coefficients
of $(1+X^{2})^{|S|}$ can be computed by a closed formula. We can
thus divide $[b]_{k}=[q]_{k}$ by $[(1+X^{2})^{|S|}]_{k}$ via truncated
polynomial division (Lemma~\ref{lem:PolyDiv}) to obtain $[q']_{k}$,
whose $k$-th coefficient counts the $k$-defect matchings in $\mathcal{N}$,
as in the idealized setting discussed before.
\end{proof}
Using a combination of truncated polynomial division (Lemma~\ref{lem:PolyDiv})
and interpolation, we compute the coefficients of $[p]_{k}$ with
oracle access for $\pDefMatch$. This completes the reduction from
$\pRestrDefMatch$ to $\pDefMatch$.
\begin{lem}
\label{lem: dual match gives polynomial}We can compute $[p]_{k}$
by a Turing fpt-reduction to $\pDefMatch$ such that all queries have maximum parameter $k$.\end{lem}
\begin{proof}
For $\xi$ with $0\leq\xi\leq k$, let $f_{\xi}\in\mathbb{Z}[X]$
be the evaluation of the expression $f_{\ell}$ defined in Lemma~\ref{lem: weighted-matching-sum}
at $\ell=\xi$. Define $p_{\xi}^{(k)}\in\mathbb{Z}[X]$ by
\begin{equation}
p_{\xi}^{(k)}:=\left[\mu(G_{S,\xi})/f_{\xi}\right]_{k}.\label{eq: definition eval-poly}
\end{equation}

\begin{claim}
\label{claim: eval p truncated}We have $p_{\xi}^{(k)}=[p(\cdot,\xi)]_{k}=[p]_{k}(\cdot,\xi)$.\end{claim}
\begin{proof}
The first identity holds by the definition of $p$ in (\ref{eq:poly-p}),
and by the definition of $p_{\xi}^{(k)}$. The second identity holds
because, for all $t\in\mathbb{N}$, the coefficient of $X^{t}$ in
$p$ is a polynomial in $\ell$ and does not depend on $X$. Hence
we may arbitrarily interchange (i) the operation of substituting $\ell$
by expressions not depending on $X$ (and by numbers $\xi\in\mathbb{N}$
in particular), and (ii) the operation of truncating to the first
$k$ coefficients.
\end{proof}
Recall that $a_{t}\in\mathbb{Z}[\ell]$ for $t\in\mathbb{N}$ denotes
the coefficient of $X^{t}$ in $p$, which has degree at most $k$
(in the indeterminate $\ell$) by Fact~\ref{fact: monomials of p}.
Hence, for fixed $t\in\mathbb{N}$, if we knew the values $a_{t}(0),\ldots,a_{t}(k)$,
we could recover the coefficients of $a_{t}\in\mathbb{Z}[\ell]$ via
univariate polynomial interpolation.
But for $0\leq\xi,t\leq k$, we can obtain the value $a_{t}(\xi)$
as the coefficient of $X^{t}$ in $p_{\xi}^{(k)}$. This follows from
Claim~\ref{claim: eval p truncated}. It remains to compute the polynomials
$p_{0}^{(k)},\ldots,p_{k}^{(k)}$ with an oracle for $\pDefMatch$:
First, we observe that the constant coefficient in $f_{\xi}$ is $1$
for all $0\leq\xi\leq k$, so we can apply the definition of $p_{\xi}^{(k)}$
from (\ref{eq: definition eval-poly}) and truncated polynomial division
(Lemma~\ref{lem:PolyDiv}) to compute $p_{\xi}^{(k)}$ from $[\mu(G_{S,\xi})]_{k}$
and $f_{\xi}$. 

It remains to compute $[\mu(G_{S,\xi})]_{k}$ and $f_{\xi}$.
Note that the coefficients of $f_{\xi}$ admit a closed expression
by definition, and that $[\mu(G_{S,\xi})]_{k}$ can be computed by
querying the oracle for $\pDefMatch$ to obtain the number of matchings
in $G_{S,\xi}$ with $0,\ldots,k$ defects.
\end{proof}
We recapitulate the proof of Theorem \ref{main thm: defect-match} in the following.
\begin{proof}[Proof of Theorem~\ref{main thm: defect-match}]
By Theorem~\ref{thm: ApexPerfMatch is hard}, the problem $\pApexPerfMatch$
is $\sharpWone$-hard, and we have reduced it to $\pRestrDefMatch$
in Lemma~\ref{lem: hardness of restricted dual matchings}. By Lemma~\ref{lem: dual match gives polynomial},
we can use oracle calls to $\pDefMatch$ with maximum parameter $k$
to compute the polynomial $[p]_{k}$, and by Lemma~\ref{lem: polynomial gives restricted dual matchings},
the coefficients of $[p]_{k}$ allow to recover the solution to $\pRestrDefMatch$
in polynomial time. These two steps establish the second reduction
in (\ref{eq:reduction-chain}). 

Note that both reductions incur only linear blowup on the parameter. 
Hence, the lower bound of $n^{\Omega(k / \log k)}$ for $\pApexPerfMatch$ under $\sharpETH$ from Theorem~\ref{thm: ApexPerfMatch is hard} carries over to $\pDefMatch$.
\end{proof}

\section{\label{sec:Algorithm-for-restricted-apices}Apices with few adjacent
faces}

We prove Theorem~\ref{thm: apex} and give an FPT-algorithm
for a restricted version of the problem $\PerfMatch$ on graphs $G$
with an apex set $A$ of size $k$ such that every apex can see only
a bounded number of faces. To this end, we first prove a stronger
version of Theorem~\ref{thm: face-defects} that allows us to compute
$\MatchSum(G)$ rather than just count matchings in $G$.
\begin{thm}
\label{thm: MatchSum-boundedfaces}Assume we are given a drawing of
a planar graph $G$ with vertex-weights $w:V(G)\to\mathbb{Q}$ and
faces $F_{1},\ldots,F_{s}$ for $s\in\mathbb{N}$ such that all vertices
$v\in V(G)$ with $w(v)\neq0$ satisfy $v\in V(F_{1})\cup\ldots\cup V(F_{s})$.
Then we can compute $\MatchSum(G)$ in time $\mathcal{O}(2^{s}\cdot n^{3})$.
\end{thm}
\begin{proof}
We first create a partition $B_{1},\ldots,B_{s}$ of $\bigcup_{i\in[s]}V(F_{i})$
such that $B_{i}\subseteq F_{i}$ for $i\in[s]$ and $B_{i}\cap B_{j}=\emptyset$
for $i\neq j$. This can be achieved trivially by assigning each vertex
that occurs in several faces $F_{i}$ to some arbitrarily chosen set
$B_{i}$. 

Now we define a type $\theta_{M}\in\{0,1\}^{s}$ for each $M\in\M[G]$.
For $i\in[s]$, we define 
\[
\theta_{M}(i):=\begin{cases}
1 & |\usat(M)\cap B_{i}|\,\mathrm{odd},\\
0 & |\usat(M)\cap B_{i}|\,\mathrm{even.}
\end{cases}
\]
For $\theta\in\{0,1\}^{s}$, let $\M_{\theta}[G]$ denote the set
of matchings $M\in\M[G]$ with $\theta_{M}=\theta$, and define 
\[S_{\theta}=\sum_{M\in\M_{\theta}[G]}\prod_{v\in\usat(M)}w(v).\]

It is clear that $\MatchSum(G)=\sum_{\theta\in\{0,1\}^{s}}S_{\theta}$.
We show how to compute $S_{\theta}$ for fixed $\theta$ in time $\mathcal{O}(n^{3})$
by reduction to $\PerfMatch$ in planar graphs. For this argument,
we momentarily define $\MatchSum(G)$ on graphs that have vertex-
and edge-weights $w:V(G)\cup E(G)\to\mathbb{Q}$:
\[
\MatchSum(G)=\sum_{M\in\M[G]}\left(\prod_{v\in\usat(M)}w(v)\right)\left(\prod_{e\in M}w(e)\right).
\]

As shown in the proof of Theorem~3.3 in \cite{Valiant2008}, and in Example~15 in \cite{DBLP:journals/corr/Curticapean15a}, for every $t\in\mathbb{N}$, there exist explicit
planar graphs $D_{t}^{0}$ and $D_{t}^{1}$ with $\mathcal{O}(t)$ vertices, which contain special vertices $u_{1},\ldots,u_{t}$
such that all of the following holds:
\begin{enumerate}
\item The graphs $D_{t}^{0}$ and $D_{t}^{1}$ can be drawn in the plane
with $u_{1},\ldots,u_{t}$ on their outer faces.
\item Let $H$ be a vertex- and edge-weighted graph with distinct vertices $X=\{v_{1},\ldots,v_{t}\}\subseteq V(H)$
and let $H'$ be obtained from $H$ by placing a disjoint copy of
$D_{t}^{0}$ into $H$ and connecting $v_{i}$ to $u_{i}$ with an
edge of weight $w(v_{i})$ for all $i\in[t]$. Assign weight $0$
to the vertices $v_{i}$ and to all vertices of $D_{t}^{0}$. Then 
\begin{equation}
\MatchSum(H')=\sum_{\substack{M\in\M[H]\\
|\usat(M)\cap X|\,\mathrm{even}
}
}\left(\prod_{v\in\usat(M)}w(v)\right)\left(\prod_{e\in M}w(e)\right)
\label{eq: matchsumH'}
\end{equation}

\item The above statement also applies for $D_{t}^{1}$, but the corresponding sum in (\ref{eq: matchsumH'})
ranges over those $M\in\M[H]$ where $|\usat(M)\cap X|$ is odd
rather than even.
\end{enumerate}
We observe that inserting $D_{t}^{0}$ or $D_{t}^{1}$ into the face
of a planar graph preserves planarity. Hence, we can insert $D_{|B_{i}|}^{\theta(i)}$
at the vertices $B_{i}$ along face $F_{i}$ in $G$, for each $i\in[s]$,
and obtain a planar graph $G_{\theta}$. By construction, we have
$\MatchSum(G_{\theta})=S_{\theta}$. Furthermore, all vertex-weights
in $G_{\theta}$ are $0$ by construction, so we actually have $\MatchSum(G_{\theta})=\PerfMatch(G_{\theta})$.
Since $G_{\theta}$ is planar, we can evaluate $\PerfMatch(G_{\theta})$
in time $\mathcal{O}(n^{3})$, thus concluding the proof.
\end{proof}

Note that the above theorem allows us to recover the number of $k$-defect matchings in $G$ that have all defects on fixed distinguished faces,
for any $k\in\mathbb{N}$: Let $G_{X}$ be
obtained from $G$ by assigning weight $X$ to each vertex. Then $p :=\MatchSum(G_{X})$
is a polynomial of degree at most $n$ and can be interpolated from
evaluations $p(0),\ldots p(n)$, but each of these evaluations can
be computed in time $\mathcal{O}(2^{s}\cdot n^{3})$ by Theorem~\ref{thm: MatchSum-boundedfaces}.
As we know, the $k$-th coefficient of $p(X)$ is equal to the number
of $k$-defect matchings in $G$.

In the following, we extend this argument by using a variant of multivariate
polynomial interpolation (Lemma~\ref{lem: multivar-interpolation})
that applies when we do not require the values of \emph{all} coefficients,
but rather only those in a ``slice'' of total degree $k$, for fixed
$k\in\mathbb{N}$. Here, the polynomial $p$ to be interpolated features
a distinguished indeterminate $X$, and we wish to extract the coefficient
$a_{k}$ of $X^{k}$, which is in turn a polynomial. Under certain
restrictions, this can be achieved with $f(k)\cdot n$ evaluations,
where $n$ denotes the degree of $X$ in $p$.
\begin{lem}
\label{lem: sliced grid interpolation} Let $p\in\mathbb{Z}[X,\lambda]$
be a multivariate polynomial in the indeterminates $X$ and $\vec{\lambda}=(\lambda_{1},\ldots,\lambda_{t})$.
Consider $p\in(\mathbb{Z}[\vec{\lambda}])[X]$ and assume that $p$
has degree $n$ in $X$, and that for all $s\in\mathbb{N}$, the coefficient
$a_{s}\in\mathbb{Z}[\vec{\lambda}]$ of $X^{s}$ in $p$ has total
degree at most $s$. Let $k\in\mathbb{N}$ be a given parameter, and
let $\Xi=\Xi_{0}\times\ldots\times\Xi_{t}\subseteq\mathbb{Q}^{t+1}$
with $|\Xi_{0}|=n+1$ and $|\Xi_{i}|=k+1$ for all $i>0$. Then we
can compute the coefficients of the polynomial $a_{k}\in\mathbb{Z}[\vec{\lambda}]$
with $\mathcal{O}(|\Xi|^{3})$ arithmetic operations when given as
input the set $\{(\xi,p(\xi))\mid\xi\in\Xi\}$.
\end{lem}
\begin{proof}
We consider the grid $\Xi'$ defined by removing the first component
from $\Xi$, that is, $\Xi'=\Xi_{1}\times\ldots\times\Xi_{t}.$
Observe that $p(\cdot,\xi')\in\mathbb{Z}[X]$ holds for $\xi'\in\Xi'$.
Write $\Xi_{0}=\{c_{0},\ldots,c_{n}\}$ and note that, for
fixed $\xi'\in\Xi'$, our input contains all evaluations 
\[
p(c_{0},\xi'),\ldots,p(c_{n},\xi'),
\]
so we can use univariate interpolation
to determine the coefficient of $X^{k}$ in $p(\cdot,\xi')$. This
coefficient is equal to $a_{k}(\xi')$ by definition. 
By performing this process for all $\xi'\in\Xi'$, we can evaluate
$a_{k}(\xi')$ on all $\xi'\in\Xi'$, and hence interpolate the polynomial
$a_{k}\in\mathbb{Z}[\vec{\lambda}]$ via grid interpolation (Lemma~\ref{lem: multivar-interpolation}).
\end{proof}

This brings us closer to the proof of Theorem~\ref{thm: apex}. To
proceed, we first consider the case that $A$ is an independent
set; the full algorithm is obtained by reduction to this case.
\begin{lem}
\label{M:Apex-Algo: lem: Algo with independent apices}Let $G$ be
an edge-weighted graph, given as input together with an independent
set $A\subseteq V(G)$ of size $k$, a planar drawing of $H=G-A$,
and faces $F_{1},\ldots,F_{s}$ that contain all neighbors of $A$.
Then we can compute $\PerfMatch(G)$ in time $k^{\mathcal{O}(2^{k})}\cdot2^{\mathcal{O}(s)}\cdot n^{4}$.\end{lem}
\begin{rem}
\label{M:Apex-Algo: rem: Unit-weight edges}We may assume that every
edge $av\in E(G)$ with $a\in A$ and $v\in V(G)\setminus A$ has
weight $1$: Otherwise, replace $av$ by a path $ar_{1}r_{2}v$ with
fresh vertices $r_{1},r_{2}$, together with edges $ar_{1}$ and $r_{1}r_{2}$
of unit weight, and an edge $r_{2}v$ of weight $w(e)$. This clearly
preserves the apex number, the value of $\PerfMatch$, and ensures
that every apex is only incident with unweighted edges.\end{rem}
\begin{proof}
Recall that $\DM_{k}[H]$ denotes the set of $k$-defect matchings
in $H$. By Remark~\ref{M:Apex-Algo: rem: Unit-weight edges}, we
can assume that all edges incident with $A$ have unit weight. Let
\[
\mathcal{C}=\{M\in\DM_{k}[H]\mid\usat(M)\subseteq N_{G}(A)\}.
\]
Given any matching $M\in\mathcal{C}$, let $t(M)$ denote its \emph{type}\footnote{Please note that these types have no connection to those used in the proof of Theorem~\ref{thm: MatchSum-boundedfaces}.}, 
which is defined as the following \emph{multiset} with precisely
$k$ elements from $2^{A}$: 
\[
t(M)=\{N_{G}(v)\cap A\mid v\in\usat(M)\}.
\]
For the set of all such types, we write $\mathcal{T}=\{t(M)\mid M\in\mathcal{C}\}$
and observe that $|\mathcal{T}|\leq(2^{k})^{k}=2^{k^{2}}$. For $t\in\mathcal{T}$,
define a graph $S_{t}$ as follows: Create an independent set $[k]$,
corresponding to $A$. Then, for each $N\in t$, create a vertex $v_{N}$
that is adjacent to all of $N\subseteq[k]$. We note that every \emph{perfect}
matching $M\in\PM[G]$ can be decomposed uniquely as $M=B(M)\dot{\cup}I(M)$
with a $k$-defect matching $B(M)\in\mathcal{C}$ and a perfect matching
$I(M)\in\PM[S_{t(B(M))}]$. That is, $B(M)=M-A$ and $I(M)=M[A\cup\usat(B(M))]$.
For $t\in\mathcal{T}$, let 
\begin{eqnarray*}
\mathcal{C}_{t} & = & \{M\in\mathcal{C}\mid t(M)=t\},\\
P_{t} & := & \sum_{N\in\mathcal{C}_{t}}\prod_{e\in N}w(e).
\end{eqnarray*}
It is clear that $\{\mathcal{C}_{t}\}_{t\in\mathcal{T}}$ partitions
$\mathcal{C}$, and this implies 
\begin{equation}
\PerfMatch(G)=\sum_{t\in\mathcal{T}}P_{t}\cdot\PerfMatch(S_{t}).\label{M:eq: Bounded-face apex, Base+Interface matching}
\end{equation}
To see this, note that each perfect matching of type $t$ can be obtained
by extending some matching $M\in\mathcal{C}_{t}$ (all of which have
$k$ defects) by a perfect matching from $\usat(M)$ to $A$, which
is precisely a perfect matching of $S_{t}$. Note that we require
here that edges between $\usat(M)$ and $A$ have unit weight, otherwise
the graphs $S_{t}$ would have to be edge-weighted as well and might
no longer depend on $t$ only, but would also have to incorporate
the edge-weights of $G$.

Since $|E(S_{t})|\leq k^{2}$, we can compute $\PerfMatch(S_{t})$
in time $2^{\mathcal{O}(k^{2})}$ by brute force for each $t\in\mathcal{T}$.
Hence, we can use (\ref{M:eq: Bounded-face apex, Base+Interface matching})
to determine $\PerfMatch(G)$ in time $|\mathcal{T}|\cdot2^{\mathcal{O}(k^{2})}$
if we know $P_{t}$ for all $t\in\mathcal{T}$. In the remainder of
this proof, we show how to compute $P_{t}$ by using multivariate
polynomial interpolation and the algorithm for $\MatchSum$ presented
in Theorem~\ref{thm: MatchSum-boundedfaces}. To this end, define
indeterminates $\vec{\lambda}=\{\lambda_{R}\mid R\subseteq A\}$ corresponding
to subsets of the apices. Let $X$ denote an additional distinguished
indeterminate, and define the following polynomial $p\in\mathbb{Z}[X,\vec{\lambda}]$.
In this definition, we abbreviate $w(M):=\prod_{e\in M}w(e)$. 
\begin{equation}
p(X,\vec{\lambda}):=\sum_{M\in\mathcal{C}}w(M)\cdot X^{|\usat(M)|}\cdot\prod_{v\in\usat(M)}\lambda_{N_{G}(v)\cap A}.\label{eq: poly-p}
\end{equation}

For each type $t\in\mathcal{T}$, say $t=\{N_{1},\ldots,N_{k}\}$,
the coefficient of $X^{k}\cdot\lambda_{N_{1}}\cdot\ldots\cdot\lambda_{N_{k}}$
in $p$ is equal to $P_{t}$. Hence, we can extract $P_{t}$ for all
$t\in\mathcal{T}$ from the coefficients of the monomials in $p$
that have degree exactly $k$ in $X$. Let us denote these monomials
by $\mathfrak{N}$, and observe that each monomial $\nu\in\mathfrak{N}$
has total degree $k$ in $\vec{\lambda}$ by the definition of $p$ in
(\ref{eq: poly-p}). 

If we can evaluate $p$ on the elements $(r,\xi)$ from the grid $\Xi=[n+1]\times[k+1]^{2^{|A|}}$,
then we can compute the coefficients of all $\nu\in\mathfrak{N}$
in $p$, and thus $P_{t}$ for all $t\in\mathcal{T}$, by sliced grid
interpolation (Lemma~\ref{lem: sliced grid interpolation}). Note
that $|\Xi| \leq \mathcal{O}(n\cdot k^{2^{k}})$. We compute these evaluations
$p(r,\xi)$ as $p(r,\xi)=\MatchSum(H')$, where the vertex-weighted
graph $H'=H'(r,\xi)$ is obtained from $H$ via the weight function
\[
w(v):=\begin{cases}
0 & \mbox{if }v\notin N_{G}(A),\\
r\cdot\xi_{N_{G}(v)\cap A} & \mbox{otherwise.}
\end{cases}
\]

Since all vertices with non-zero weight in $H'$ are contained
in the faces $F_{1},\ldots,F_{s}$, we can compute $\MatchSum(H')$
in time $\mathcal{O}(2^{s}\cdot n^{3})$ with Theorem~\ref{thm: MatchSum-boundedfaces}.
We obtain the values $P_{t}$ for all $t\in\mathcal{T}$, so we obtain
$\PerfMatch(G)$ via (\ref{M:eq: Bounded-face apex, Base+Interface matching})
in the required time.
\end{proof}
It remains to lift Lemma~\ref{M:Apex-Algo: lem: Algo with independent apices}
to the case that $A$ is not an independent set. This follows easily
from the fact that, whenever $E(G)=E\dot{\cup}E'$, then every perfect
matching $M\in\PM[G]$ must match every vertex $v\in V(G)$ into exactly
one of the sets $E$ or $E'$.
\begin{proof}[Proof of Theorem~\ref{thm: apex}]
Let $\mathcal{A}=\mathcal{M}[G[A]]$ denote the set of (not necessarily
perfect) matchings of the induced subgraph $G[A]$, and note that
$|\mathcal{A}|\leq2^{k^{2}}$. For $M\in\mathcal{A}$, let $a_{M}=\PerfMatch(G_{M})$,
where $G_{M}$ is defined by keeping from $A$ only $\usat(M)$, and
then deleting all edges between the remaining vertices of $A$. We can compute
$a_{M}$ by Lemma~\ref{M:Apex-Algo: lem: Algo with independent apices},
since the remaining part of $A$ in $G_{M}$ is an independent set.
It is also easily verified that $\PerfMatch(G)=\sum_{M\in\mathcal{A}}a_{M}\cdot\prod_{e\in M}w(e)$,
so we can compute $\PerfMatch$ as a linear combination of $2^{k^{2}}$
values, each of which can be computed by Lemma~\ref{M:Apex-Algo: lem: Algo with independent apices}.\end{proof}

\section*{Acknowledgments}

The author wishes to thank D\'{a}niel Marx and Holger Dell for pointing out the
connection between counting perfect matchings in $k$-apex graphs and $\pDefMatch$
during the Dagstuhl Seminar 10481 on Computational Counting in 2010.
Furthermore, thanks to Mingji Xia for interesting discussions about this topic; in particular, Theorem~\ref{thm: face-defects} was found in joint work on combined
signatures back in 2013. Thanks also to Markus Bl\"aser for reading earlier drafts of this material as it appeared in my PhD thesis, and thanks to the reviewers of this version for providing helpful comments.

\bibliography{p278-curticapean}

\end{document}